\documentclass[11pt,a4paper,reqno]{amsart}%
\usepackage[utf8]{inputenc}
\usepackage{mathrsfs}
\usepackage{dsfont}
\usepackage{hyperref}
\usepackage{amsmath}
\usepackage{amssymb}
\usepackage{amsthm}
\usepackage{amsfonts}
\usepackage{amstext}
\usepackage{amsopn}
\usepackage{amsxtra}
\usepackage{mathrsfs}
\usepackage{dsfont}
\usepackage{esint}
\usepackage{graphicx}
\newtheorem{theorem}{Theorem}

\newtheorem{remark}{Remark}

\newtheorem{lemma}{Lemma}
\newtheorem{corollary}{Corollary}

\newcommand{\ii}{\infty}
\newcommand\R{{\ensuremath {\mathbb R} }}

\newcommand\1{{\ensuremath {\mathds 1} }}
\renewcommand\phi{\varphi}
\newcommand{\gH}{\mathfrak{H}}
\newcommand{\gS}{\mathfrak{S}}

\newcommand{\wto}{\rightharpoonup}

\newcommand{\cB}{\mathcal{B}}

\newcommand{\cH}{\mathcal{H}}

\renewcommand{\epsilon}{\varepsilon}

\newcommand{\norm}[1]{ \left| \! \left| #1 \right| \! \right| }
\newcommand{\tr}{{\rm Tr}}

\renewcommand{\le}{\leqslant}
\renewcommand{\geq}{\geqslant}
\renewcommand{\leq}{\leqslant}

\title[A family of monotone quantum relative entropies]{A family of monotone quantum relative entropies}

\author[M. Lewin]{Mathieu LEWIN}
\address{CNRS \& Universit\'e de Cergy-Pontoise, Mathematics Department (UMR 8088), F-95000 Cergy-Pontoise, France} 
\email{mathieu.lewin@math.cnrs.fr}

\author[J. Sabin]{Julien SABIN}
\address{Universit\'e de Cergy-Pontoise, Mathematics Department (UMR 8088), F-95000 Cergy-Pontoise, France} 
\email{julien.sabin@u-cergy.fr}

\date{\today}

\begin{document}

\begin{abstract}
We study here the elementary properties of the relative entropy $\cH(A,B)=\tr[\phi(A)-\phi(B)-\phi'(B)(A-B)]$ for $\phi$ a convex function and $A,B$ bounded self-adjoint operators. In particular, we prove that this relative entropy is monotone if and only if $\phi'$ is operator monotone. We use this to appropriately define $\cH(A,B)$ in infinite dimension.
\end{abstract}

\thanks{\copyright\,2013 by the authors. This paper may be reproduced, in its entirety, for non-commercial purposes.}

\maketitle



In this paper we introduce a relative entropy $\cH(A,B)$ of two bounded operators $A$ and $B$ on a Hilbert space and discuss its elementary properties. Our relative entropy $\cH(A,B)$ could be useful in many practical situations including quantum information theory. In a separate work~\cite{LewSab-13a}, we will use it to discuss the orbital stability of certain stationary states for the Hartree equation of an interacting gas containing infinitely many quantum particles.

For simplicity we work on the interval $I=[0,1]$ but all our results can be generalized to any finite interval. We consider two $N\times N$ hermitian matrices $A,B$ whose spectrum is a subset of $I$. For $\phi\in C^0([0,1],\R)\cap C^1((0,1),\R)$ a \emph{convex} function, we introduce the following relative entropy\footnote{This expression always makes sense if $0<B<1$. If the spectrum of $B$ contains $0$ or $1$ and if $\phi$ is not differentiable at these points, then we let $\cH(A,B)=+\ii$ except when $A=B$ on $\ker B$, $\ker(1-B)$, or $\ker B\oplus\ker(1-B)$. In this case the trace is by definition taken on the orthogonal to these subspaces.}
\begin{equation}
\boxed{\cH(A,B)=\tr\Big(\phi(A)-\phi(B)-\phi'(B)(A-B)\Big).}
\label{eq:def_relative_entropy} 
\end{equation}

We remind the reader of Klein's lemma.

\begin{lemma}[Klein~{\cite[Prop. 3.16]{OhyPet-93},~\cite[Thm 2.5.2]{Ruelle}}]\label{lem:Klein}
If $f_k,g_k:[a,b]\to\R$ are $K$ measurable functions such that $\sum_{k=1}^K f_k(x)g_k(y)\geq0$ for all $x,y\in [a,b]$, then we have
$\sum_{k=1}^K \tr\,\big( f_k(A)g_k(B)\big)\geq0$ for all hermitian matrices $a\leq A,B\leq b$.
\end{lemma}

Because of our assumption that $\phi$ is convex, we obtain from the lemma that the trace in~\eqref{eq:def_relative_entropy} is always non-negative. If $\phi$ is strictly convex, it can be proved that $\cH(A,B)=0$ only when $A=B$. Therefore, $\cH(A,B)$ is an appropriate object to measure how $A$ is close to $B$.

The formula~\eqref{eq:def_relative_entropy} is motivated by several physical situations. First, the usual von Neumann entropy
$$\cH_{\text{vN}}(A,B)=\tr A\big(\log A-\log B\big)$$
corresponds to taking 
$$\phi_{\rm vN}(x)=x\log(x),\qquad x\in [0,1].$$ 
This is not the only function of interest, however. If we take two quasi-free states $\alpha$ and $\beta$ on the CAR (resp. CCR) algebra of a (finite-dimensional) Hilbert space $\gH$, then the corresponding von Neumann entropy can be expressed in terms of the one-particle density matrices $A$ and $B$ of $\alpha$ and $\beta$, leading to the expression~\eqref{eq:def_relative_entropy} with, this time,
$$\phi_{\text{CAR}}(x)=x\log(x)+(1-x)\log(1-x),\qquad x\in [0,1]$$
and
$$\phi_{\text{CCR}}(x)=x\log(x)-(1+x)\log(1+x),\qquad x\in [0,1],$$
respectively~\cite{BacLieSol-94}. It seems therefore natural to study the properties of $\cH$ for a general convex function $\phi$, and this is the main purpose of this paper. We will particularly study the possibility to define $\cH$ for two bounded operators in an \emph{infinite} dimensional Hilbert space $\gH$. \emph{Monotonicity} then plays a crucial role and we discuss it in the next section.

\section{Monotonicity}

Our goal is to define $\cH(A,B)$ for $A$ and $B$ two bounded operators on an infinite dimensional Hilbert space. The usual technique employed in statistical mechanics is to first project the system into a finite dimensional space and then to take the limit when the dimension goes to $+\ii$. This means that we will define
\begin{equation}
\cH(A,B):=\lim_{k\to\ii}\cH(P_kAP_k,P_kBP_k)
\label{eq:def_limit} 
\end{equation}
where $P_k$ is an increasing sequence of finite-rank orthogonal projectors on $\gH$. For a general (convex) function $\phi$, the limit is not guaranteed to exist. Furthermore, the limit could be different for two sequences $P_k$ and $P'_k$. One natural way to ensure a good behavior is to add the constraint that $\cH$ is \emph{monotone}, that is, satisfies the inequality 
\begin{equation}
\boxed{\cH(PAP,PBP)\leq \cH(A,B) }
\label{eq:monotonicity}
\end{equation}
for every finite self-adjoint matrices $0\leq A,B\leq 1$ and every orthogonal projector $P$. The limit on the right of~\eqref{eq:def_limit} is then monotone and $\cH(A,B)$ is well defined in $\R^+\cup\{+\ii\}$. Our purpose in this section is to characterize the convex functions $\phi$ for which $\cH$ is monotone in the sense of~\eqref{eq:monotonicity}.

Let us emphasize that another notion of monotonicity is often used in the literature. Assume that our (finite-dimensional) Hilbert space is the tensor product $\gH=\gH_1\otimes\gH_2$. Define $\Phi(A)=\tr_{\gH_2}(A)$, the partial trace with respect to the second Hilbert space. Then the von Neumann relative entropy (corresponding to $\phi=\phi_{\rm vN}$) satisfies that 
\begin{equation}
\cH_{\rm vN}(\Phi(A),\Phi(B))\leq \cH_{\rm vN}(A,B) 
\label{eq:strong_monotonicity}
\end{equation}
for every $0\leq A,B\leq 1$ with $\tr(A)=\tr(B)$. Indeed, the property holds for any trace-preserving completely positive map $\Phi$ between two Hilbert spaces~\cite{OhyPet-93}. Since $A\mapsto PAP+P^\perp AP^\perp$ is completely positive and trace-preserving, it is easy to deduce that the von Neumann entropy also satisfies~\eqref{eq:monotonicity}. Furthermore, it can be shown (see~\cite[Chap. 2]{OhyPet-93}) that $\phi(x)=ax\log(x)+bx+c$ with $a\geq0$ are the only functions for which~\eqref{eq:strong_monotonicity} is satisfied.

The weaker monotonicity property~\eqref{eq:monotonicity} is associated with a decomposition of the Hilbert space $\gH$ into a direct sum $\gH=\gH_1\oplus\gH_2$, instead of a tensor product. Indeed, if we have a quasi-free state on the CCR (resp. CAR) algebra of a Hilbert space $\gH$, the restriction of this state to a sub-algebra of observables in $\gH_1\subset\gH$ is obtained by taking a \emph{partial trace for the state} and by \emph{replacing the one-particle density matrix $A$ by $P_1AP_1$}. The tensor structure at the level of states is transformed into a direct sum at the level of density matrices. From this discussion we conclude that the relative entropy associated with $\phi_{\rm CCR/CAR}$ (defined above) must also satisfy the monotonicity property~\eqref{eq:monotonicity}, even if they do not satisfy the stronger monotonicity~\eqref{eq:strong_monotonicity}.

In this section we want to characterize the convex functions $\phi$ for which $\cH$ is monotone in the sense of~\eqref{eq:monotonicity}. In a recent paper~\cite{AudHiaPet-10}, Audenaert, Hiai and Petz have studied the strong sub-additivity of the generalized entropy $A\mapsto -\tr\;\phi(A)$ when $\phi'$ is operator monotone. Petz then conjectured that the operator monotonicity of $\phi'$ is necessary. Here we express everything in terms of the relative entropy~\eqref{eq:def_relative_entropy} and we prove Petz's conjecture in this setting.

\begin{theorem}[Monotonicity]\label{thm:monotonicity}
Let $\phi\in C^0([0,1],\R)\cap C^1((0,1),\R)$ be a convex function. Then the following are equivalent

\begin{enumerate}
\item $\phi'$ is operator monotone on $(0,1)$;

\smallskip

\item For any linear map $X:\gH_1\to\gH_2$ on finite-dimensional spaces $\gH_1$ and $\gH_2$ with $X^*X\leq1$, and for any $0\leq A,B\leq 1$ on $\gH_1$, we have 
\begin{equation}
\cH(XA X^*,XBX^*)\leq \cH(A,B),
\label{eq:monotonicity_2}
\end{equation}
with $\cH(A,B)$ defined in~\eqref{eq:def_relative_entropy}.
\end{enumerate}
\end{theorem}

Typical functions $\phi$ satisfying the conditions in Theorem~\ref{thm:monotonicity} are provided in~\cite{AudHiaPet-10} and include
$$\phi(x)=\begin{cases}
(t+x)\log(t+x),&t\geq0\\
-\log(t+x),&t>0,\\
-x^m,&0<m\leq 1,\\
x^m,&1\leq m\leq 2.\\
\end{cases}$$
By using L\"owner's characterization of operator-monotone functions, it is possible to write an integral representation formula for the allowed functions $\phi$. This is done later in the proof, see~\eqref{eq:Lowner2}.

\begin{proof}
As a preliminary we remark that $x\mapsto x\phi'(x)$ possesses a limit when $x\to0^+$. Indeed, since $\phi$ is convex by assumption, then $x\mapsto x\phi'(x)-\phi(x)$ must be increasing on $(0,1)$. Also, $x\phi'(x)-\phi(x)\geq -\phi(0)$ and we conclude that $x\mapsto x\phi'(x)-\phi(x)$ possesses a limit when $x\to0^+$. Since $\phi$ is continuous by assumption, then $x\phi'(x)$ must have a limit as well, and this limit is $\geq0$. The same argument shows that $(1-x)\phi'(x)$ has a limit when $x\to 1^-$.

Next we prove that (2) implies (1).
Let $A$ be any matrix such that $0<\eta<A<1-\eta$, and $P$ be any orthogonal projector. We use the shorthand notation $A_{11}=PAP$, $A_{22}=P^\perp AP^\perp$ and so on. Let $X$ be any self-adjoint matrix such that $P^\perp X=XP^\perp=P^\perp$ and $0< X\leq1$. Then we have by assumption
\begin{align*}
\cH\Big(XA X,XA_{11}X+A_{22}\Big)&=\cH\Big(XA X,X(A_{11}+A_{22})X\Big)\\
&\leq \cH\big(A,A_{11}+A_{22}\big). 
\end{align*}
Note that 
\begin{multline*}
\tr_\gH\; \phi'\big(A_{11}+A_{22}\big)\big(A-A_{11}-A_{22}\big)\\
=\tr_{P\gH\oplus P^\perp \gH}
\begin{pmatrix}
\phi'\big(A_{11}\big) & 0\\
0 & \phi'\big(A_{22}\big)
\end{pmatrix}
\begin{pmatrix}
0 & A_{12}\\
A_{21} & 0
\end{pmatrix}=0.
\end{multline*}
A similar argument for the term involving $X$ gives
\begin{align*}
 0&\leq \cH\big(A,A_{11}+A_{22}\big) - \cH\Big(XA X,XA_{11}X+A_{22}\Big)\\
&=\tr_\gH\Big(\phi(A)-\phi(A_{11}+A_{22})-\phi(XAX)+\phi(XA_{11}X+A_{22}) \Big)\\
&=\tr_\gH\Big(\phi(A)-\phi(XAX)\Big)-\tr_{P\gH}\Big(\phi(A_{11})-\phi((XAX)_{11}) \Big).
\end{align*}
Now we choose $X:=\exp(-\epsilon C)$ for some $C>0$ living in the range of $P$ and $\epsilon>0$. We remark that $XAX=A-\epsilon(CA+AC)+O(\epsilon^2)$. Using the differentiability of $A\mapsto \tr\phi(A)$, we obtain 
$$\tr_\gH\Big(\phi(A)-\phi(e^{-\epsilon C}Ae^{-\epsilon C})\Big)=2\epsilon \tr_\gH \big(CA\phi'(A)\big)+o(\epsilon).$$
The same argument for the other term and the fact that the inequality is valid for all $C>0$ on $P\gH$ and all $\epsilon>0$ gives
$$\big(A\phi'(A)\big)_{11}\geq A_{11}\phi'(A_{11}).$$
This inequality is valid for all $0<A<1$. We have already seen that $x\mapsto x\phi'(x)$ can be extended by continuity at $x=0$. So the property stays true for all $0\leq A< 1$. By~\cite[Thm V.2.3 \& Thm V.2.9]{Bhatia}, this exactly proves that the function $\phi'$ is operator monotone on $I$. These results give us the additional information that $\lim_{x\to0^+}x\phi'(x)\leq0$, and since we already had the other inequality we deduce that $\lim_{x\to0^+}x\phi'(x)=0$.

Next we turn to the proof that (1) implies (2), for which we follow arguments from~\cite{AndPet-09,AudHiaPet-10}. By L\"owner's theorem (see~\cite[Corollary V.4.5]{Bhatia}), we can write $\phi'$ as
\begin{equation}
\phi'(x)=a+b\int_{-1}^1 \frac{2x-1}{1-\lambda(2x-1)}\,d\nu(\lambda)
\label{eq:Lowner} 
\end{equation}
with $b\geq0$ and $\nu$ a Borel probability measure on $[-1,1]$. 
Integrating this formula, we find
$$\phi(x)=ax+c-\frac{b}2\int_{-1}^1\left(\frac{2x-1}{\lambda}+\frac{\log\big(1+\lambda (1-2x)\big)}{\lambda^2}\right)\,d\nu(\lambda).$$
The fact that $\phi$ possesses a limit at $0$ and $1$ implies that $\nu(\{-1\})=\nu(\{1\})=0$. From this we can actually prove that 
\begin{equation}
\lim_{x\to0^+}x\phi'(x)=\lim_{x\to1^-}(1-x)\phi'(x)=0.
\label{eq:limits_phi} 
\end{equation}
After letting $t=-(1+\lambda)/(2\lambda)$ for $\lambda\in(-1,0)$ and $t=(1-\lambda)/(2\lambda)$ for $\lambda\in(0,1)$, $\phi$ can also be written as
\begin{multline}
\phi(x)=a'x+c'-\int_{0}^\ii\left(\log(t+x)-\log(t+1/2)-\frac{2x-1}{2t+1}\right)\,d\nu_1(t)\\
-\int_{0}^\ii\left(\log(t+1-x)-\log(t+1/2)+\frac{2x-1}{2t+1}\right)\,d\nu_2(t)
\label{eq:Lowner2}
\end{multline}
with $\nu_1$ and $\nu_2$ two positive Borel measures satisfying 
$$\int_{0}^\ii \frac{\big(d\nu_1+d\nu_2\big)(t)}{(2t+1)^2}<\ii.$$
The constraint that $\phi$ admits limits at $0$ and $1$ now means that
\begin{equation}
-\int_0^1\log(t)\big(d\nu_1+d\nu_2\big)(t)<\ii 
\label{eq:measure_finite_log}
\end{equation}
and it implies $\nu_1(\{0\})=\nu_2(\{1\})=0$, as we have already said before. The parameters are given in the following table for the usual distributions.

\begin{center}
\begin{tabular}{|c|c|c|c|c|}
\hline
 & $d\nu_1/dt$&$d\nu_2/dt$&$a'$&$c'$\\
\hline
von Neumann & 1 & 0 & $1-\log2$ &$-1/2$ \\
Fermi-Dirac & 1 & 1 & 0 &$-\log2$\\
Bose-Einstein & $\1(0\le t\le1)$ & 0 & $-\log3$ & $\log2-\log3$ \\
\hline
\end{tabular}
\end{center}

It suffices to prove the monotonicity for $\phi(x)=-\log(t+x)$ and for $\phi(x)=-\log(t+1-x)$ with $t\geq0$. We start by considering the case $t=0$ and we let $\phi(x)=-(1/2)\log(x)$. We follow the method of~\cite{AndPet-09}. For $A$ a positive $n\times n$ matrix, we introduce the corresponding normalized Gaussian, defined by
$$f_A(x):=\frac{\sqrt{\det(A)}}{(2\pi)^{n/2}}\exp\left(-\frac{x^TAx}{2}\right).$$
A simple calculation shows that the \emph{classical} entropy of $f_A$ is
\begin{align*}
-\int_{\R^n}f_A(x)\log f_A(x)\,dx&=-\frac{\log\det(A)}{2}+ \frac{n\log(2\pi e)}{2}\\
&=\frac{\tr\log (A^{-1})}{2}+ \frac{n\log(2\pi e)}{2},
\end{align*}
and that the \emph{classical} relative entropy of two Gaussians $f_A$ and $f_B$ is
\begin{equation*}
\cH_{\rm class}(f_A,f_B)=\int_{\R^n}f_A(x)\log \frac{f_A(x)}{f_B(x)}\,dx=\cH(A^{-1},B^{-1}),
\end{equation*}
see~\cite[Eq. (21)--(22)]{AndPet-09}. Now, let $P$ be any orthogonal projection on $\gH$. Up to a change of variables we may assume for simplicity that $x=(x_1,x_2)$ with $Px=(x_1,0)$ for all $x\in\R^n$. The marginal of a Gaussian is computed in~\cite[Lemma 4]{AndPet-09} and it is equal to
$$\int f_A(x_1,x_2)\,dx_2=f_{A_{11}-A_{12}A_{22}^{-1}A_{21}}(x_1)$$
where we have used the decomposition
$$A=\begin{pmatrix}
A_{11}&A_{12}\\
A_{21}&A_{22}
\end{pmatrix}.$$
Noticing that $(A_{11}-A_{12}A_{22}^{-1}A_{21})^{-1}=(A^{-1})_{11}$ by the Schur complement formula, we deduce that
$$\cH_{\rm class}\left(\int f_{A^{-1}} dx_2,\int f_{B^{-1}}dx_2\right)=\cH\big(A_{11},B_{11}\big)=\cH(PA P,PB P).$$
The inequality $\cH(PA P,PB P)\leq \cH(A,B)$ now follows from the well-known monotonicity of the classical relative entropy in terms of marginals. The argument is exactly the same if $\log(x)$ is replaced by $\log(t+x)$ and $\log(t+1-x)$, since $(t+A)_{11}=t+A_{11}$ and $(t+1-A)_{11}=t+1-A_{11}$.

We have proved the monotonicity for a projection $P$. If $U:\gH_1\to\gH_2$ is a partial isometry between two spaces such that $U^*U=1_{\gH_1}$ and $UU^*\leq 1_{\gH_2}$, it is obvious that $\cH(UAU^*,UBU^*)=\cH(A,B)$ since $U\gH_1$ is then isometric to a subspace of $\gH_2$. On the other hand, if we have $UU^*=1_{\gH_2}$ and $U^*U\leq 1_{\gH_1}$, then $\cH(UAU^*,UBU^*)\leq \cH(A,B)$ by the previous result on projections, since $U^*\gH_2$ is now isometric to a subspace of $\gH_1$.

In order to deduce the result for an arbitrary operator $X:\gH\to\gH$ with $X^*X\leq1$, we double the dimension of the space, that is, we introduce the partial isometry 
$$\begin{array}{rccrcl}
U:&\gH&\to&\gH&\oplus&\gH\\
&f&\mapsto&Xf&\oplus &\sqrt{1-X^*X}f
\end{array}$$
and we infer $\cH(UA U^*,UB U^*)= \cH(A,B)$, since $U^*U=1_{\gH}$.
Note that
$$UAU^*=\begin{pmatrix}
XA X & XA\sqrt{1-X^*X}\\
\sqrt{1-X^*X}A X & \sqrt{1-X^*X}A\sqrt{1-X^*X}\\
\end{pmatrix}.$$
If we project onto the first Hilbert space of $\gH\oplus\gH$, this decreases the relative entropy and we obtain the result. The case of $X:\gH_1\to\gH_2$ with different Hilbert spaces and $X^*X\leq1$ follows from the polar decomposition $X=UY$ with $U$ a partial isometry and $0\leq Y=\sqrt{X^*X}\leq1$. This concludes our proof that (1) implies (2).
\end{proof}

After having characterized the functions $\phi$ for which $\cH$ is monotone, we quickly mention the consequences on the entropy $S(A)=-\tr\;\phi(A)$, for completeness. 

One of the fundamental properties of the von Neumann entropy (defined with $\phi_{\rm vN}(x)=x\log(x)$) is its strong subadditivity (SSA), which was proved by Lieb and Ruskai in \cite{LieRus-73a,LieRus-73b}. The fact that SSA is very important for large quantum systems was first remarked by Robinson and Ruelle in \cite{RobRue-67}, see also \cite{Wehrl-78} and~\cite{HaiLewSol_2-09}. It is well-known that SSA is equivalent to the monotonicity of the relative von Neumann entropy $\cH$ under completely positive trace-preserving maps~\cite{OhyPet-93,Carlen-10,Wehrl-78}, as in~\eqref{eq:strong_monotonicity}.

The usual definition of SSA is expressed using a decomposition of the ambient Hilbert space $\gH$ into a tensor product of the form $\gH=\gH_1\otimes\gH_2\otimes\gH_3$, and we do not recall it here. In our setting there is another concept of SSA, which is associated with a decomposition of $\gH$ into a direct sum $\gH=\gH_1\oplus\gH_2\oplus\gH_3$ of three spaces (instead of a tensor product), and which was considered before in~\cite{AudHiaPet-10}. For $\phi=\phi_{\rm CCR/CAR}$, this is just the usual SSA for quasi-free states, expressed in terms of the corresponding one-particle density matrices.
We denote by $P_j$ the corresponding orthogonal projections, as well as $P_{12}=P_1+P_2$ and $P_{123}=1$. 

\begin{corollary}[SSA for the entropy]
We assume that $\phi\in C^0([0,1],\R)$ and that $\phi'$ is operator monotone on $(0,1)$.
Let $S(A)=-\tr_\gH\,\phi(A)$, with $\dim(\gH)<\ii$. Then we have 
$$S(P_{123}AP_{123})+S(P_2AP_2)\leq S(P_{12}AP_{12})+S(P_{23}AP_{23}),$$
for all self-adjoint matrices $0\leq A\leq 1$ on $\gH$.
\end{corollary}

\begin{proof}
By using the monotonicity of the relative entropy, we deduce that 
$\cH(P_{123}AP_{123},P_{23}AP_{23})\geq\cH(P_{12}AP_{12},P_{2}AP_{2})$.
It is straightforward to verify that $\cH(P_{123}AP_{123},P_{23}AP_{23})=S(P_{123}AP_{123})-S(P_{23}AP_{23})$ and that $\cH(P_{12}AP_{12},P_{2}AP_{2})=S(P_{12}AP_{12})-S(P_{2}AP_{2})$ and the result follows.
\end{proof}

\begin{remark}
Since an operator monotone function is always $C^\ii$, there is no need to assume anymore that $\phi\in C^1$.
\end{remark}

\section{Definition in infinite-dimensional spaces}

We now turn to the definition of $\cH(A,B)$ in an infinite-dimensional space $\gH$. For any two self-adjoint operators $0\leq A,B\leq 1$, we let
\begin{equation}
\cH(A,B):=\lim_{k\to\ii}\cH(P_kA P_k,P_kB P_k)
\label{eq:def_relative_entropy_abstract}
\end{equation}
where $P_k$ is any increasing sequence of finite-dimensional projections on $\gH$, such that $P_k\to 1$ strongly.
The following result says that the limit exists in $[0,\ii]$ and provides important properties of $\cH$.

\begin{theorem}[Generalized relative entropy in infinite dimension]\label{thm:relative_entropy}
We assume that $\phi\in C^0([0,1],\R)$ and that $\phi'$ is operator monotone on $(0,1)$.

\smallskip

\noindent$\bullet$ \emph{($\cH$ is well-defined).} For any increasing sequence $P_k$ of finite-dimensional projections on $\gH$ such that $P_k\to 1$ strongly, the sequence 
$\cH(P_kA P_k,P_kB P_k)$ is monotone and possesses a limit in $\R^+\cup\{+\ii\}$. This limit does not depend on the chosen sequence $P_k$ and hence $\cH(A,B)$ is well-defined in $\R^+\cup\{+\ii\}$.

\smallskip

\noindent$\bullet$ \emph{(Monotonicity).} The so-defined relative entropy $\cH(A,B)$ is \emph{monotone}: for any linear map $X:\gH_1\to\gH_2$ such that $X^*X\leq1$ and any $0\leq A,B\leq 1$ on $\gH_1$, we have 
\begin{equation}
\cH(A,B)\geq \cH(XA X^*,XBX^*).
\label{eq:monotonicity_3}
\end{equation}

\smallskip

\noindent$\bullet$ \emph{(Approximation).} If $X_k:\gH_1\to\gH_k$ is a sequence such that $X_k^*X_k\leq1$ and $X_k^*X_k\to 1$ strongly in $\gH_1$, then
\begin{equation}
\cH(A,B)=\lim_{k\to\ii}\cH(X_kA X_k^*,X_kBX_k^*).
\label{eq:strong_CV}
\end{equation}

\smallskip

\noindent$\bullet$ \emph{(Weak lower semi-continuity).} The relative entropy is \emph{weakly lower semi-continuous}: if $0\leq A_n,B_n\leq 1$ are two sequences such that $A_n\wto A$ and $B_n\wto B$ weakly-$\ast$ in $\cB(\gH)$, then
\begin{equation}
\cH(A,B)\leq\liminf_{n\to\ii} \cH(A_n,B_n).
\label{eq:wlsc}
\end{equation}
\end{theorem}

Most of the statements of the theorem easily follow from the monotonicity in finite dimension proved in Theorem~\ref{thm:monotonicity}, showing the importance of this concept.

\begin{proof}
 We split the proof into several steps.

\subsubsection*{Step 1. Lower semi-continuity in finite dimension}

We want to prove that $\cH(A,B)\leq\liminf_{n\to\ii}\cH(A_n,B_n)$ if $A_n\to A$ and $B_n\to B$ are convergent sequences of hermitian matrices. All the terms in the definition of $\cH(A,B)$ are continuous, except possibly for $(A,B)\mapsto \tr\, \phi'(B) (A-B)$. We write this term as
\begin{multline*}
\tr\, \phi'(B) (B-A)\\= -\tr\, \phi'(B)\1(\phi'(B)\leq0)A+\tr\, \phi'(B)\1(\phi'(B)\geq0)(1-A)\\
+\tr\, B\phi'(B)\1(\phi'(B)\leq0)
-\tr\, (1-B)\phi'(B)\1(\phi'(B)\geq0).
\end{multline*}
Using that the first two terms involve non-negative matrices and that the next ones are continuous by~\eqref{eq:limits_phi}, we obtain the result.

\subsubsection*{Step 2. The definition in infinite dimension does not depend on $(P_k)$}
Let us consider two sequences $(P_k)$ and $(P'_\ell)$ of increasing orthogonal projections, such that $P_k\to1$ and $P'_\ell\to1$ strongly. By the monotonicity in finite dimension, we have 
$\cH(P_kP'_\ell A P'_\ell P_k,P_kP'_\ell B P'_\ell P_k)\leq \cH(P'_\ell AP'_\ell,P'_\ell BP'_\ell)$.
Taking first $\ell\to\ii$ using that $P_kP'_\ell A P'_\ell P_k\to P_kP'_\ell B P'_\ell P_k$ and the lower semi-continuity of $\cH$ in the finite dimensional space ${\rm Ran}(P_k)$, we get
$$\cH(P_kA P_k,P_kB P_k)\leq \lim_{\ell\to\ii} \cH(P'_\ell A P'_\ell,P'_\ell B P'_\ell).$$
Taking then $k\to\ii$ shows that 
$$\lim_{k\to\ii}\cH(P_kA P_k,P_kB P_k)\leq \lim_{\ell\to\ii} \cH(P'_\ell A P'_\ell,P'_\ell B P'_\ell).$$
Exchanging $P_k$ and $P'_\ell$ gives the other inequality.

\subsubsection*{Step 3. Monotonicity in infinite dimension}
Let $X:\gH_1\to\gH_2$ be such that $X^*X\leq1_{\gH_1}$. For fixed sequences $(P_k)$ and $(P'_\ell)$ of finite-rank projections on $\gH_1$ and $\gH_2$ respectively, we have 
\begin{align*}
\cH(P'_\ell XP_k A P_kX^*P'_\ell,P'_\ell XP_kB P_kX^*P'_\ell)&\leq \cH(XP_k A P_kX^*,XP_kB P_kX^*)\\
&\leq \cH(P_k A P_k,P_kB P_k). 
\end{align*}
Taking first $k\to\ii$ using that $P'_\ell XP_k A P_kX^*P'_\ell\to P'_\ell X A X^*P'_\ell$ with the finite rank projection $P'_\ell$ fixed, and only then $\ell\to\ii$, gives the monotonicity.

\subsubsection*{Step 4. Weak lower semi-continuity}
If $ A_n\wto A$ and $B_n\wto B$ weakly-$\ast$ in $\cB(\gH)$, then 
$$\cH(P_k A P_k,P_kB P_k)\leq\liminf_{n\to\ii}\cH(P_k A_n P_k,P_kB_n P_k)\leq \liminf_{n\to\ii}\cH( A_n,B_n)$$
for any fixed sequence $(P_k)$ as before. Here we have used that $P_k A_n P_k\to P_k A P_k$ for a finite-rank projection and the lower semi-continuity in finite dimension. Taking $k\to\ii$ gives the weak lower semi-continuity. 

\subsubsection*{Step 5. Convergence for any sequence $X_k\to1$}
Let $X_k:\gH\to\gH$ be any sequence of self-adjoint operators, such that $X_k^*X_k\leq1$ and $X_k^*X_k\to1$ strongly. We want to prove that $\cH(X_k A X_k^*,X_kB X_k^*)\to \cH(A,B)$ but, by the polar decomposition $X_k=U_kY_k$, we can always assume that $X_k$ is self-adjoint, hence that $X_k\to1$ strongly. We have $\cH(X_k A X_k,X_kB X_k)\leq \cH(A,B)$. On the other hand, $X_k A X_k\to  A$ strongly and thus the weak lower semi-continuity implies the reverse inequality. If $X_k:\gH_1\to\gH_2$, the argument is similar.
\end{proof}

\section{Klein inequalities and consequences}

In this last section, we derive some useful bounds on $A-B$ when the relative entropy $\cH(A,B)$ is finite. For instance, we will prove that $A-B$ must be a Hilbert-Schmidt (hence compact) operator. Our main result is inspired of~\cite[Thm 1]{HaiLewSei-08} and of~\cite[Lemma 1]{FraHaiSeiSol-12} and it is mainly based on Klein's Lemma~\ref{lem:Klein}, hence the name ``Klein inequalities''.

\begin{theorem}[Klein inequalities]\label{thm:Klein}
We assume that $\phi\in C^0([0,1],\R)$, and that $\phi'$ is operator monotone on $(0,1)$ and not constant.

\smallskip

\noindent$\bullet$ \emph{(Lower bound).} There exists $C>0$ (depending only on $\phi$) such that 
\begin{equation}
\cH(A,B)\geq C\,\tr\,\big(1+|\phi'(B)|\big)(A-B)^2,
\label{eq:Klein}
\end{equation}
for all $0\leq A,B\leq 1$. Hence, $A-B$ is Hilbert-Schmidt when $\cH(A,B)<\ii$.

\smallskip

\noindent$\bullet$ \emph{(Upper bound).} Similarly, we have
\begin{equation}
\cH(A,B)\leq C\,\tr\,\left(\frac{1}{B^{2}}+\frac{1}{(1-B)^2}\right)(A-B)^2,
\label{eq:Klein_upper}
\end{equation}
for all $0\leq A,B\leq1$.

\smallskip

\noindent$\bullet$ \emph{(Regularity).} If $0\leq A\leq1$ and $0<\epsilon\leq A',B\leq 1-\epsilon$ are such that $A-A'$ and $A'-B$ are Hilbert-Schmidt operators, then
\begin{equation}
\big|\cH(A,B)-\cH(A',B)\big|\leq C_\epsilon\Big(\norm{A-A'}_{\gS^2}^2+\norm{A'-B}_{\gS^2}\norm{A-A'}_{\gS^2}\Big)
\label{eq:Klein_Lipschitz}
\end{equation}
where $C_\epsilon$ depends on $\phi$ and on $\epsilon$.
\end{theorem}

We remark that the inequality~\eqref{eq:Klein_upper} is far from optimal if $\phi(x)$ is very smooth at $0$ or $1$. Similarly,~\eqref{eq:Klein_Lipschitz} holds without the assumptions that the density matrices have their spectrum in $[\epsilon,1-\epsilon]$, if $\phi''$ is bounded on $[0,1]$.

\begin{proof}[Proof of Theorem \ref{thm:Klein}]
We have 
\begin{equation*}
-\frac{\log(t+x)-\log(t+y)-\frac{x-y}{t+y}}{(x-y)^2}=\frac{1}{(t+y)^{2}}\,\frac{h-\log(1+h)}{h^2},
\end{equation*}
with $h=(x-y)/(t+y)$. Using that $h\mapsto h^{-2}(h-\log(1+h)$ is decreasing, we deduce that
$$\min_{0\leq x\leq 1}\left(-\frac{\log(t+x)-\log(t+y)-\frac{x-y}{t+y}}{(x-y)^2}\right)=\frac{\frac{1-y}{t+y}-\log\left(1+\frac{1-y}{t+y}\right)}{(1-y)^2}.$$
Using now that
$$\frac{h-\log\left(1+h\right)}{h^2}\geq \frac{1}{2(h+1)}$$
for all $h\geq0$, we conclude that
\begin{equation}
-\left(\log(t+x)-\log(t+y)-\frac{x-y}{t+y}\right)\\ \geq \frac{(x-y)^2}{2(1+t)(t+y)}.
\label{eq:Klein_x_y_1}
\end{equation}
By replacing $x$ by $1-x$ and $y$ by $1-y$, we get the same estimate with $t+1-y$ in place of $t+y$ in the denominator on the right side. Inserting in~\eqref{eq:Lowner2} we find that, for all $0<x,y<1$,
\begin{align*}
&\phi(x)-\phi(y)-\phi'(y)(x-y)\\
&\qquad \geq \frac{b}2\left(\int_0^\ii \frac{d\nu_1(t)}{(1+t)(t+y)}+\int_0^\ii \frac{d\nu_2(t)}{(1+t)(t+1-y)}\right)(x-y)^2\\
&\qquad \geq C\left(\int_1^\ii \frac{d\nu_1(t)+d\nu_2(t)}{(2t+1)^2}+\int_0^1 \frac{d\nu_1(t)}{t+y}+\int_0^1 \frac{d\nu_2(t)}{t+1-y}\right)(x-y)^2.
\end{align*}
From the integral representation of $\phi$ in~\eqref{eq:Lowner2}, we see that 
$$\phi'(x)\underset{x\to0^+}{\sim}-b\int_0^1\frac{d\nu_1(t)}{(t+x)}+C,\qquad \phi'(x)\underset{x\to1^-}{\sim}b\int_0^1\frac{d\nu_2(t)}{(t+1-x)}+C.$$
Therefore we have proved that
\begin{equation*}
\forall 0\leq x,y\leq1,\quad  \phi(x)-\phi(y)-\phi'(y)(x-y)\geq C(1+|\phi'(y)|)(x-y)^2.
\end{equation*}
The result follows for finite matrices by Klein's Lemma~\ref{lem:Klein}, and for general operators by the definition of $\cH$ using projections on finite-dimensional spaces.

Maximizing instead of minimizing in~\eqref{eq:Klein_x_y_1}, we find
\begin{align*}
-\left(\log(t+x)-\log(t+y)-\frac{x-y}{t+y}\right) &\leq (x-y)^2\frac{-\frac{y}{t+y}-\log\left(1-\frac{y}{t+y}\right)}{y^2}\\
&\leq (x-y)^2\frac{1-\log\left(\frac{t}{t+y}\right)}{(t+y)^2}
\label{eq:Klein_x_y_2}
\end{align*}
and a similar estimate for $1-x$ and $1-y$. After integrating with respect to $t$, we see that the integral converges for large $t$, uniformly in $y$. The only possibly diverging terms are
$$\int_0^1\frac{\log\left(\frac{t}{t+y}\right)}{(t+y)^2}d\nu_1(t)+\int_0^1\frac{\log\left(\frac{t}{t+1-y}\right)}{(t+1-y)^2}d\nu_2(t)$$
and they can always be estimated by a constant times $y^{-2}+(1-y)^{-2}$, by~\eqref{eq:measure_finite_log}. (This is far from optimal if $\nu_1$ or $\nu_2$ vanish sufficiently fast at $0$).

Now we turn to the proof of~\eqref{eq:Klein_Lipschitz}, which only has to be done in finite dimension, by the definition of $\cH$. We remark that
$$\cH(A,B)-\cH(A',B)=\cH(A,A')+\tr \big(\phi'(A)-\phi'(B')\big)(A-A').$$
Since $\phi''$ is bounded on $[\epsilon,1-\epsilon]$, we have 
$$\tr \big(\phi'(A)-\phi'(B')\big)^2\leq C\tr \big(A'-B\big)^2,$$ 
by Klein's Lemma~\ref{lem:Klein}.
By H\"older's inequality this gives
$$\left|\tr \big(\phi'(A)-\phi'(B')\big)(A-A')\right|\leq C\norm{A'-B}_{\gS^2}\norm{A-A'}_{\gS^2}.$$
Using~\eqref{eq:Klein_upper} and the assumption that $\epsilon\leq A'\leq1-\epsilon$ to estimate $\cH(A,A')$ gives~\eqref{eq:Klein_Lipschitz}.
\end{proof}

Using the previous Klein inequalities, we can get some information on the set of one-particle density matrices $A$ which have a finite relative entropy with respect to a given $B$. Of particular interest is fact that any operator $A$ such that $\cH(A,B)<\ii$ can be approximated in an appropriate sense by a sequence of finite rank perturbations of $B$.

\begin{corollary}[Density of finite rank perturbations]\label{cor:density}
We assume that $\phi\in C^0([0,1],\R)$ and that $\phi'$ is operator monotone on $(0,1)$ and that $0\leq A,B\leq1$ have a finite relative entropy, $\cH(A,B)<\ii$. Then for every $\epsilon>0$, there exists $0\leq A'\leq1$ such that
\begin{itemize}
\item $A'-A$ has a finite rank and its eigenvectors belong to $\ker(B)$, to $\ker(1-B)$, or to the domain of $B^{-2}+(1-B)^{-2}$ in the orthogonal of these two subspaces;
\item $\tr (1+|\phi'(B)|)(A-A')^2\leq \epsilon$;
\item $\big|\cH(A,B)-\cH(A',B)|\leq \epsilon$.
\end{itemize}
\end{corollary}

\begin{proof}
We have several cases to look at, depending whether $\phi'$ diverges at the end points $0$ and $1$ and depending whether $B$ has the end points as eigenvalues or not. First we note that if $\phi'$ diverges at $0$ and $\ker B\neq\{0\}$, then $A=B$ on $\ker B$ and therefore there is nothing to do on this subspace and we can argue on the orthogonal of $\ker B$.
So let us now assume that $\ker B=\ker(1-B)=\{0\}$. We introduce the spectral projector $\Pi_\eta:=\1(\eta\leq B\leq 1-\eta)$. By the spectral theorem, we have $\Pi_\eta\to1$ strongly when $\eta\to0$. By Theorem~\ref{thm:relative_entropy}, we deduce that
$$\lim_{\eta\to0}\cH(\Pi_\eta A\Pi_\eta,\Pi_\eta B\Pi_\eta)=\lim_{\eta\to0}\cH(B+\Pi_\eta(A-B)\Pi_\eta,B)=\cH(A,B).$$
In the first equality we have used that
\begin{align*}
\cH(B+\Pi_\eta(A-B)\Pi_\eta,B)&=
\cH(\Pi_\eta^\perp B\Pi_\eta^\perp+\Pi_\eta A\Pi_\eta,\Pi_\eta^\perp B\Pi_\eta^\perp+\Pi_\eta B\Pi_\eta)\\
&=\cH(\Pi_\eta A\Pi_\eta,\Pi_\eta B\Pi_\eta)+\cH(\Pi_\eta^\perp B\Pi_\eta^\perp,\Pi_\eta^\perp B\Pi_\eta^\perp)\\
&=\cH(\Pi_\eta A\Pi_\eta,\Pi_\eta B\Pi_\eta),
\end{align*}
since $\Pi_\eta$ commutes with $B$. By Klein's inequality~\eqref{eq:Klein}, we know that $(1+|\phi'(B)|)^{1/2}(A-B)$ is a Hilbert-Schmidt operator,  and it is then clear that 
$$\lim_{\eta\to0}\tr (1+|\phi'(B)|)\Big(\Pi_\eta(A-B)\Pi_\eta-(A-B)\Big)^2=0.$$
The density matrix $A_\eta:=B+\Pi_\eta(A-B)\Pi_\eta$ satisfies all the desired properties, except for the finite rank one, and we need to approximate it further.

In order to approximate $A_\eta$, we will work in the space $\gH_\eta:=\Pi_\eta\gH$. We first remark that, by Klein's inequality, $\Pi_\eta A\Pi_\eta$ is a Hilbert-Schmidt perturbation of $\Pi_\eta B$ and, therefore, the essential spectrum of $\Pi_\eta A\Pi_\eta$ (restricted to $\gH_\eta$) is included in the interval $[\eta,1-\eta]$. The end points $0$ and $1$ can only be isolated eigenvalues of finite multiplicity. Let us denote by $P_0$ and $P_1$ the associated projections (if they exist). Then $A'_\eta:=\Pi_\eta A\Pi_\eta+\epsilon P_0-\epsilon P_1$ is a finite rank perturbation of $A_\eta$ and a simple calculation shows that
\begin{multline*}
\cH(\Pi_\eta A\Pi_\eta+\epsilon P_0-\epsilon P_1,\Pi_\eta B)-\cH(\Pi_\eta A\Pi_\eta,\Pi_\eta B)=\big(\phi(\epsilon)-\phi(0)\big)\tr P_0\\
+ \big(\phi(1-\epsilon)-\phi(1)\big)\tr P_1+\epsilon\tr\big(\Pi_\eta \phi'(B) (P_1-P_0)\big)
\end{multline*}
(the computation must be done in a finite sequence of dimensional spaces $V_k\nearrow \gH_\eta$ and it is simpler to assume that the eigenvectors of $\Pi_\eta A\Pi_\eta$ corresponding to the end points belong to these spaces).
This goes to $0$ as $\epsilon\to0$ with $\eta$ fixed. At this step we have approximated $\Pi_\eta A\Pi_\eta$ by an operator $A'_\eta$ which has its spectrum away from $0$ and $1$, and which is a Hilbert-Schmidt perturbation of $\Pi_\eta B$. Now we can take any sequence $Q_k$ of smooth finite rank operators which converges to $A'_\eta-\Pi_\eta B$ in the Hilbert-Schmidt norm. For large $k$ the operator $A_k:=\Pi_\eta B+Q_k$ satisfies $0\leq A_k\leq1$ and its relative entropy converges to that of $A'_\eta$, by~\eqref{eq:Klein_Lipschitz} in Theorem~\ref{thm:Klein}. This concludes the proof in the case where $\ker B=\ker(1-B)=\{0\}$.

When $0$ or $1$ is an eigenvalue of $B$ and $\phi'$ is bounded at these points, the argument is exactly the same and we leave it to the reader. The projection $P_\eta$ must include all (if the multiplicity if finite) or a sequence of eigenvectors (if the multiplicity is infinite) at the end points.
\end{proof}

%
%

\bigskip

\noindent\textbf{Acknowledgements.} 
We acknowledge financial support from the European Research Council under the European Community's Seventh Framework Programme (FP7/2007-2013 Grant Agreement MNIQS 258023), and from the French ministry of research (ANR-10-BLAN-0101).


\end{document}